\documentclass[conference]{IEEEtran}

\usepackage{amssymb}
\usepackage{verbatim}
\usepackage[final]{graphicx}

\usepackage[cmex10]{amsmath}
\usepackage{cite}
\usepackage{subfigure}

\newtheorem{theorem}{Theorem}
\newtheorem{corollary}{Corollary}

\newtheorem{proof}{Proof}
\newtheorem{proposition}[corollary]{Proposition}

\begin{document}
%
\title{On Achievable Rates of the Two-user Symmetric Gaussian Interference Channel}

\author{\IEEEauthorblockN{Omar Mehanna, John Marcos and Nihar Jindal}
\IEEEauthorblockA{ University of Minnesota\\
\{meha0006, marc0312, nihar \}@umn.edu }
}

\maketitle

\begin{abstract}
We study the Han-Kobayashi (HK) achievable sum rate for the two-user symmetric Gaussian interference channel.  We find the optimal power split ratio between
the common and private messages (assuming no time-sharing), and derive a closed form expression for the corresponding sum rate. This provides
a finer understanding of the achievable HK sum rate, and allows for precise comparisons between this sum rate and that of orthogonal signaling.  One surprising finding is that despite the fact that the channel is symmetric, allowing for asymmetric power split ratio at both users (i.e., asymmetric rates) can improve the sum rate significantly. Considering the high SNR regime, we specify the interference channel value above which the sum rate achieved using asymmetric power splitting outperforms the symmetric case.
\end{abstract}

\section{Introduction}

The subject of this paper is the interference channel, which is one of the most fundamental models for wireless communication systems.
Although this model is of utmost importance, the capacity region for even the simplest two user symmetric Gaussian interference channel is not
yet fully characterized, except for the special cases of strong interference \cite{Kobayashi, Sato} and very weak interference \cite{Veeravalli,Khandani,Kramer2}.
The best known achievability strategy for the remaining unsolved cases was proposed by Han and Kobayashi (HK) in \cite{Kobayashi}, and it
combines the ideas of time-sharing and rate-splitting (i.e., dividing the transmitted message into two parts: a common part decodable by both receivers, and a private part
decodable only by the intended receiver).
Significant progress towards the general capacity region was made in \cite{Etkin}, where a new capacity upper bound was derived and
it was shown that the HK rate region, restricted to Gaussian inputs and not allowing for time-sharing, comes within one bit of that upper bound.

In this paper we build on \cite{Etkin} and more carefully study the HK achievable rate region.
In \cite{Etkin} the authors chose the private message power to be received at the noise level of the unintended receiver; although
this power split is not optimal, it is sufficient to achieve the one bit result.  On the other hand, we characterize the
exact power split that maximizes the HK achievable sum rate.  By finding the optimal power split, we are able to obtain a closed
form expression for the corresponding maximum HK sum rate.  In turn, this leads to a more precise understanding of the achievable
HK rate.  We are able to use our results to make comparisons between the HK achievable rate (without time-sharing) and the sum rate of
orthogonal signaling (i.e., TDMA/FDMA), and to exactly identify the regions where simple orthogonal signaling outperforms
rate-splitting.

Since the channel is symmetric, one would assume that it is optimal to use symmetric power splitting ratios (i.e., user 1 and user 2 use the same power splitting ratio). However, we surprisingly find that allowing for asymmetric power splits can enhance the sum rate. Specifically, if one user sends only a common message while the second user uses a specific private/common split, the sum rate outperforms that of the symmetric case for a wide range of signal and interference powers. Next, we consider the high SNR regime and precisely identify the interference channel value above which the rate achieved using the asymmetric power splitting outperforms the symmetric case. We further consider specific time sharing schemes and we show that the advantage of using such time sharing schemes, as opposed to the case of not allowing time sharing, is quite small.

\section{Network Model}
\begin{figure}
\centering
\includegraphics[width=6cm]{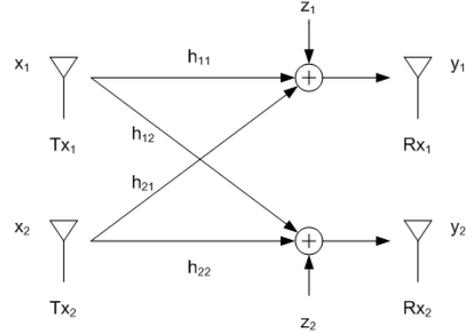}
\caption{Two-user Gaussian interference channel}
\label{GIFC}
\end{figure}

As shown in Figure \ref{GIFC}, the two-user Gaussian interference channel is modeled as:
\begin{equation}
	y_1 = h_{11}x_1+h_{21}x_2+\bar{z}_1,  \quad 	y_2 = h_{12}x_1+h_{22}x_2+\bar{z}_2
\end{equation}
where the noise processes $\bar{z}_i$ are assumed to be circularly symmetric Gaussian random variables with variance $N_0$. The transmitted signal by each user $x_i$ is subject to an average power constraint $P_i$. We also assume that the codebooks used are generated using i.i.d. random Gaussian variables as considered in \cite{Etkin}. We consider the normalized symmetric two-user Gaussian interference channel, specifically:
\begin{equation}
     |h_{11}| = |h_{22}|, \;  |h_{12}| = |h_{21}|,  \;   P_1 = P_2 = P.
\end{equation}
Thus, the channel can be expressed in standard form as:
\begin{equation}
	y_1 = x_1+\sqrt{a} x_2+z_1,  \quad 	y_2 = \sqrt{a} x_1+ x_2+z_2
\end{equation}
where $a = \frac{|h_{21}|^2 }{|h_{11}|^2 } = \frac{|h_{12}|^2 }{|h_{22}|^2}$ and $ z_i \sim \mathcal{C} N(0,1)$. Hence, the channel considered is fully characterized by the parameters $P$ (which represents the SNR) and the interference coefficient $a$.
Since the capacity is already known for the strong interference case  (i.e., $a\geq1$) \cite{Kobayashi}, \cite{Sato}, we focus
exclusively on the range $0 < a < 1$.

Each user chooses a power splitting ratio $\lambda_i$ ($0 \leq \lambda_i \leq 1$) and transmits a private
message $u_i$ with power $P_{u_i} = \lambda_i P_i$ and a common message $w_i$ with power $P_{w_i} = \bar{\lambda_i} P_i$, where $\bar{\lambda_i} = 1 - \lambda_i$.  The common message should be decodable by both receivers.  The channel can be considered as two virtual 3-users multiple access channels,
where the first has $(u_1,w_1,w_2)$ as inputs and $y_1$ as output with $u_2$ treated as noise, and the second has $(u_2,w_2,w_1)$ as inputs, $y_2$ as output, and $u_1$ is treated as noise.
In general, the HK strategy proposed in \cite{Kobayashi} allows for arbitrary splits of each user's transmit power into the private and common messages in addition to time sharing between multiple of such splits. We first consider the restricted case of no time sharing, then the case of allowing time sharing is considered later. Throughout the paper we use the notation $\gamma(x) \triangleq \log_2(1+x)$.

\section{Optimized HK Rate}

\subsection{No Time Sharing}
We first derive an expression for the maximum sum-rate in the HK region assuming no time sharing. 
For fixed $\lambda_1$ and $\lambda_2$, the HK sum-rate is defined as: 
\begin{equation}
 R_{HK}(\lambda_1,\lambda_2) \triangleq \max_{R_{1u},R_{1w},R_{2u},R_{2w}} (\underbrace{R_{1u} + R_{1w}}_{R_1} + \underbrace{R_{2u}+ R_{2w}}_{R_2})
\end{equation}
where $R_{iu}$ is the rate of a private message decoded at receiver $i$, and similarly $R_{iw}$ is the rate of a common message decoded at receiver $i$.

\begin{proposition}
For fixed $\lambda_1$ and $\lambda_2$, the HK sum-rate is:
\begin{eqnarray}\label{R_HK}
\nonumber \lefteqn{ R_{HK}(\lambda_1,\lambda_2)  =    \gamma \bigg( \frac{\lambda_1 P}{1+a\lambda_2 P}\bigg) + \gamma \bigg( \frac{\lambda_2 P}{1+a\lambda_1 P}\bigg) +}\\
   \nonumber     & &    \min \bigg\{  \gamma \bigg( \frac{a\bar{\lambda}_2 P}{1+\lambda_1 P+a\lambda_2 P}\bigg) + \gamma \bigg( \frac{a\bar{\lambda}_1 P}{1+\lambda_2 P+a\lambda_1 P}\bigg) , \\
  \nonumber      & &    \frac{1}{2} \gamma \bigg( \frac{\bar{\lambda}_1 P+a\bar{\lambda}_2 P}{1+\lambda_1 P+a\lambda_2 P}\bigg) + \frac{1}{2} \gamma \bigg( \frac{\bar{\lambda}_2 P+a\bar{\lambda}_1 P}{1+\lambda_2 P+a\lambda_1 P}\bigg) \bigg\}\\
\end{eqnarray}
Hence, the optimized HK sum-rate $R_{RS}$ is the solution to the following optimization problem:
\begin{equation}\label{SumRate}
	 R_{RS}(a,P)  \triangleq \max_{0\leq \lambda_1,\lambda_2 \leq 1} R_{HK}(\lambda_1,\lambda_2)
\end{equation}

\end{proposition}

\begin{proof}

Based on the discussions in \cite{Kobayashi, Etkin, Sason}, it can be shown that there is no rate loss if we fix the decoding order at each receiver such that both common messages are decoded first while the private message is decoded last. This is intuitive since by decoding the common message of the other user, part of the interference is canceled off. The private message will be decoded last, while the private message of the other user is treated as noise achieving a sum rate of the private messages of both users given by:
\begin{equation}\label{eqn0}
    R_{1u} + R_{2u} =  \gamma \bigg( \frac{\lambda_1 P}{1+a\lambda_2 P}\bigg) + \gamma \bigg( \frac{\lambda_2 P}{1+a\lambda_1 P}\bigg)
\end{equation}
Since the common messages from both users are decoded first while treating both private messages as interference, the sum rate of both common messages must satisfy the following constraints:
\begin{eqnarray}
 \nonumber \lefteqn{R_{1w} + R_{2w}  \leq} \\
  & & \gamma \bigg( \frac{\bar{\lambda}_1 P}{1+\lambda_1 P+a\lambda_2 P}\bigg) + \gamma \bigg( \frac{\bar{\lambda}_2 P}{1+\lambda_2 P+a\lambda_1 P}\bigg) \label{eqn1}\\
  \nonumber \lefteqn{R_{1w} + R_{2w}   \leq}\\
   & & \gamma \bigg( \frac{a\bar{\lambda}_2 P}{1+\lambda_1 P+a\lambda_2 P}\bigg) + \gamma \bigg( \frac{a\bar{\lambda}_1 P}{1+\lambda_2 P+a\lambda_1 P}\bigg) \label{eqn2}\\
  \nonumber \lefteqn{ R_{1w} + R_{2w}  \leq}\\
   \nonumber & & \frac{1}{2} \gamma \bigg( \frac{\bar{\lambda}_1 P+a\bar{\lambda}_2 P}{1+\lambda_1 P+a\lambda_2 P}\bigg) + \frac{1}{2} \gamma \bigg( \frac{\bar{\lambda}_2 P+a\bar{\lambda}_1 P}{1+\lambda_2 P+a\lambda_1 P}\bigg)\\ \label{eqn3}
\end{eqnarray}
where inequality (\ref{eqn1}) correspond to the individual rate constraint of decoding the common messages $w_1$ at receiver 1 and $w_2$ at receiver 2, inequality (\ref{eqn2}) correspond to the individual rate constraint of decoding $w_2$ at receiver 1 and $w_1$ at receiver 2, and inequality (\ref{eqn3}) correspond to the sum rate constraint of jointly decoding both common messages $w_1$ and $w_2$ at receiver 1 and receiver 2.

Comparing inequalities (\ref{eqn1}) and (\ref{eqn2}), we can see that the bound in (\ref{eqn2}) is more tight than than the bound in (\ref{eqn1}) if:
\begin{eqnarray}
\nonumber \lefteqn{ \log_2 \bigg( \frac{(1+ a\lambda_2 P + P)(1+ a\lambda_1 P + P) }{(1+\lambda_1 P+a\lambda_2 P)(1+\lambda_2 P+a\lambda_1 P)}\bigg) \geq } \\
 & & \log_2 \bigg( \frac{(1+a P + \lambda_1 P)(1+a P + \lambda_2 P)}{(1+\lambda_1 P+a\lambda_2 P)(1+\lambda_2 P+a\lambda_1 P)} \bigg)
\end{eqnarray}
or,
\begin{equation*}
    (1+P(1+ a\lambda_1))(1+P(1+ a\lambda_2)) \geq (1+P(a+ \lambda_1))(1+P(a+ \lambda_2))
\end{equation*}
which is always true for $0 \leq \lambda_i \leq 1$ and $0 < a < 1$. Therefore, combining equations (\ref{eqn0}), (\ref{eqn2}) and (\ref{eqn3}) gives:
\begin{eqnarray}
 \nonumber   \lefteqn{ \max_{R_{1u},R_{1w},R_{2u},R_{2w}} (R_{1u} + R_{1w} + R_{2u}+ R_{2w})  =}\\
 \nonumber & & \gamma \bigg( \frac{\lambda_1 P}{1+a\lambda_2 P}\bigg) + \gamma \bigg( \frac{\lambda_2 P}{1+a\lambda_1 P}\bigg) +\\
  \nonumber      & &    \min \bigg\{  \gamma \bigg( \frac{a\bar{\lambda}_2 P}{1+\lambda_1 P+a\lambda_2 P}\bigg) + \gamma \bigg( \frac{a\bar{\lambda}_1 P}{1+\lambda_2 P+a\lambda_1 P}\bigg) , \\
  \nonumber     & &  \frac{1}{2} \gamma \bigg( \frac{\bar{\lambda}_1 P+a\bar{\lambda}_2 P}{1+\lambda_1 P+a\lambda_2 P}\bigg) + \frac{1}{2} \gamma \bigg( \frac{\bar{\lambda}_2 P+a\bar{\lambda}_1 P}{1+\lambda_2 P+a\lambda_1 P}\bigg) \bigg\}\\
\end{eqnarray}
Finally, the optimized HK rate with no time sharing $R_{RS}$ is achieved by maximizing $ R_{HK}(\lambda_1,\lambda_2)$ with respect to $\lambda_1$ and $\lambda_2$.

\end{proof}


Next, we try to solve the maximization in Preposition 1. Since the channel is symmetric, it might seem that only symmetric power splits (i.e., $\lambda_1=\lambda_2=\lambda_{sym}$) need to be considered.  The following theorem characterizes performance under this restriction:

\begin{theorem}
If the power splits must satisfy $\lambda_1=\lambda_2=\lambda_{sym}$, the maximum symmetric sum rate achievable with rate splitting is:
\begin{eqnarray}\label{EqPwrRate}
\lefteqn {   R_{sym}(a,P) = \max_{0\leq \lambda_1=\lambda_2 \leq 1} R_{HK}(\lambda_1,\lambda_2) }\\
\nonumber &=&    \begin{cases}
    2 \gamma \bigg( \frac{P}{1+aP}\bigg) \qquad \qquad\text{ if }    P \leq \frac{1-a}{a^2}  \\
    \\
    2 \gamma \bigg( \frac{(a^2P+a-1)(1-a)+aP}{1+a(a^2P+a-1)}\bigg) \\
     \qquad\qquad \qquad \qquad\text{ if }  \frac{1-a}{a^2} < P \leq \frac{1-a^3}{a^3(a+1)}\\
     \\
      \gamma \bigg( \frac{1-a}{2a}\bigg) + \gamma \bigg( \frac{(1+a)^2P-(1-a)}{2}\bigg)  \\
     \qquad\qquad \qquad \qquad\text{ if } P > \frac{1-a^3}{a^3(a+1)}
    \end{cases}\\
\end{eqnarray}
and the corresponding optimal power split ratio is:
\begin{equation}\label{EqPwrLamda}
    \lambda^*_{sym} = \begin{cases}
     1  & \text{ if }    P \leq \frac{1-a}{a^2}\\
    \frac{a^2P+a-1}{P} & \text{ if }  \frac{1-a}{a^2} < P \leq \frac{1-a^3}{a^3(a+1)}\\
    \frac{1-a}{(1+a)(a P)} & \text{ if } P > \frac{1-a^3}{a^3(a+1)}
    \end{cases}
\end{equation}

\end{theorem}

\begin{proof}
Refer to Appendix A for a detailed proof.
\end{proof}
In the first region $P \leq \frac{1-a}{a^2}$, Theorem 1 shows that transmitting only a private message and treating interference as noise (i.e., $\lambda(a,P)=1$), maximizes the HK sum-rate. This is consistent with the findings of \cite{Veeravalli, Kramer2, Khandani} where it was shown that this strategy achieves capacity for the further restricted region: $P\leq \frac{1}{2}a^{-3/2}-a^{-1}$.

If we constrain one of the users to send only a common message (i.e., $\lambda_i = 0$ where $i=$ 1 or 2), the corresponding maximum sum rate achievable with such a structure $R_{asym}$ is obtained by substituting $\lambda_1=0$ and $\lambda_2= \lambda$ in equation (\ref{SumRate}), hence:
\begin{eqnarray}\label{Asym_Rate}
	\nonumber \lefteqn{R_{asym}(a,P)  =  \max_{ 0\leq \lambda \leq 1} R_{HK}(\lambda_1=0,\lambda_2= \lambda) } \\
  \nonumber &=& \max_{0\leq \lambda \leq 1} \bigg[ \gamma(\lambda P) + \min \bigg\{   \gamma \bigg( \frac{a P }{1+\lambda P}\bigg) + \gamma \bigg( \frac{a \bar{\lambda}  P }{1+a\lambda P}\bigg)  , \\
 \nonumber & & \frac{1}{2} \gamma \bigg( \frac{\bar{\lambda} P+ aP }{1+\lambda P}\bigg) + \frac{1}{2}  \gamma \bigg( \frac{P+a \bar{\lambda}  P }{1+a\lambda P}\bigg) \bigg\} \bigg]\\
    & = &  \max_{0\leq \lambda \leq 1} \min \{ \Omega_1(a , \lambda , P) \, , \, \Omega_2(a , \lambda , P)\}
\end{eqnarray}
where $\Omega_1(a , \lambda , P) = \gamma(\lambda P) + \gamma \bigg( \frac{a P }{1+\lambda P}\bigg) + \gamma \bigg( \frac{a \bar{\lambda}  P }{1+a\lambda P}\bigg)$
and  $\Omega_2(a , \lambda , P) = \gamma(\lambda P) + \frac{1}{2} \gamma \bigg( \frac{\bar{\lambda} P+ aP }{1+\lambda P}\bigg) + \frac{1}{2}  \gamma \bigg( \frac{P+a \bar{\lambda}  P }{1+a\lambda P}\bigg)$.

For $ P \geq \frac{1-a}{a^2} $, it can easily be shown that  $\frac{\partial \Omega_1 }{\partial \lambda} < 0$ (i.e, $\Omega_1$ is a decreasing function in $\lambda$) and $\frac{\partial \Omega_2}{\partial \lambda} > 0$ (i.e, $\Omega_2$ is an increasing function in $\lambda$). Thus, the solution of equation (\ref{Asym_Rate}) is achieved at $\lambda$ that satisfies $ \Omega_1 = \Omega_2$. Hence, it follows that:
\begin{equation}\label{UnEqPwrRate}
    R_{asym}(a,P)= \log_2 \bigg( \frac{(1+\lambda_{asym} P +aP)(1+aP)}{1+a \lambda_{asym} P} \bigg)
\end{equation}
where $\lambda_{asym}$, which is the power splitting ratio of the other user, is the solution to the following equation:
 \begin{equation}\label{UnEqPwrLamda}
    \sqrt{\frac{1+\lambda P}{1+a \lambda P} } (1+P+aP) = \frac{(1+\lambda P +aP)(1+aP)}{1+a \lambda P}.
\end{equation}
Otherwise for $P < \frac{1-a}{a^2}$, it can be shown that $\frac{\partial \Omega_1 }{\partial \lambda} > 0$ and $\frac{\partial \Omega_2}{\partial \lambda} > 0$, hence the solution of equation (\ref{Asym_Rate}) is obtained at $\lambda_{asym}=1$, yielding the sum rate: $R_{asym}= \log_2 (1+P+aP)$. However, it is easy to see that $R_{asym} < R_{sym}$ for this power region.

Based on the structure of $R_{HK}(\lambda_1,\lambda_2)$ for various values of $a$ and $P$, we conjecture that the maximum HK sum-rate $R_{RS}$ is achieved either using symmetric power splits (i.e., $\lambda_1=\lambda_2=\lambda_{sym}$) and maximizing $R_{HK}$ over $\lambda_{sym}$ or by constraining one of the users to send only a common message and maximizing $R_{HK}$ over the other user's power splitting ratio; i.e., $R_{RS} = \max \{ R_{sym} ,  R_{asym} \}$. The main observations that lead to this conjecture are given in Appendix B. We surprisingly find that despite the fact that the channel is symmetric, the asymmetric HK sum rate $R_{asym}$, which results in an asymmetric rate for user 1 and user 2, outperforms that of the symmetric case $R_{sym}$ for a wide range of $a$ and $P$ values.

If orthogonal signaling is used instead of rate-splitting, the resulting sum-rate is:
\begin{equation}
    R_{orth} = \gamma(2P)
\end{equation}
It worthwhile noting that the lower bound studied in \cite{Etkin} corresponds to the HK sum-rate in (\ref{SumRate}) with the suboptimal choice $\lambda_1 = \lambda_2 = \frac{1}{aP}$.  The corresponding rate, which we denote as $R_{ETW}$, is:
\begin{eqnarray}\label{R_LB}
 \nonumber    R_{ETW} &=& \min \bigg\{ \gamma\left(\frac{1}{2a}\right)+\gamma\left( \frac{P(1+a)-1}{2}\right) , \\
     & & 2 \gamma\left( \frac{1-a+a^2P}{2a}\right)\bigg\}.
\end{eqnarray}
Comparisons between $R_{sym}$, $R_{asym}$, $R_{ETW}$ and $R_{orth}$ for different $a$ and $P$ values are shown in Section \ref{numerical}.

\subsection{High SNR Results}

In order to better understand performance at high SNR, we now study the asymptotic sum-rate offset with a fixed value of $a$ and $P \rightarrow \infty$:
\begin{equation}
    \Delta R(a) \triangleq \lim_{P\rightarrow \infty} (R -  \log_2(P)).
\end{equation}
Straightforward calculation yields the following:
\begin{eqnarray}
    \Delta R_{sym}(a) &=&  \log_2 \bigg( \frac{(1+a)^3}{4 a}\bigg)\\
    \Delta R_{asym}(a) &=&  \log_2 \bigg( \frac{1+a}{\sqrt{a}}\bigg)\\
  \Delta R_{ETW}(a) &=& \log_2 \left(\frac{(2a+1)(a+1))}{4a}\right)\\
  \Delta R_{orth}(a) &=& 1
\end{eqnarray}
Comparing these asymptotic sum-rate offsets, we can conclude the following at the high SNR:
\begin{itemize}
  \item $ R_{sym} >  R_{asym}$ for $0<a<0.087$ while $ R_{asym}> R_{sym}$ for $0.087\leq a<1$.
  \item $ R_{asym} >  R_{orth}$ for all values of $a$ (i.e., $R_{orth}$ is always suboptimal at the high SNR), whereas $ R_{sym}< R_{orth}$ for $ \sqrt{5} - 2 \leq a < 1$.
  \item $ R_{sym} \geq  R_{ETW}$  for all values of $a$ which is a result of the sub-optimal choice of $\lambda$ in $R_{ETW}$.
\end{itemize}
It can be further shown that $R_{sym} $ is achieved with $\lambda_1=\lambda_2=\frac{1-a}{(1+a)(a P)}$ while $R_{asym}$ is achieved with $\lambda_{asym} = \frac{a^{3/2}}{1+a-a^{1/2}}$. Figure \ref{Comparison} gives the plots of $\Delta R_{sym}$,  $\Delta R_{asym}$,  $\Delta R_{ETW}$, and $\Delta R_{orth}$ versus $a$.

\begin{figure}
\centering
\includegraphics[height=5.5cm, width=\columnwidth]{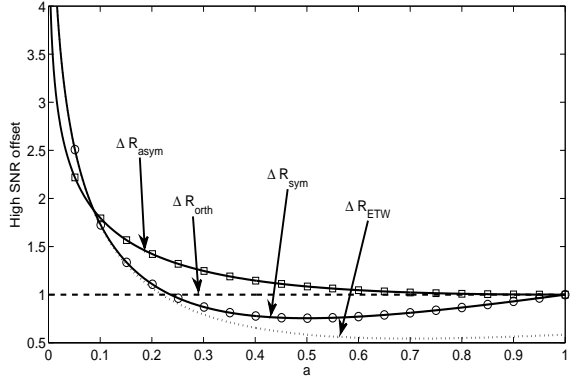}
\caption{Comparing  $\Delta R_{sym}$,  $\Delta R_{asym}$,  $\Delta R_{ETW}$,  and $\Delta R_{orth}$.}
\label{Comparison}
\end{figure}

\subsection{Allowing Time Sharing}

By allowing time sharing between multiple power splits, the total transmission time can be divided into $N$ time slots where each time slot $n$ ($n = 1, \ldots, N$) correspond to a fraction $\delta(n)$ of the whole time (i.e., $\sum_{n=1}^N \delta(n) =1$). The $i$'th user can transmit with power up to $\alpha_i(n) P$ in the $n$'th time slot (this power is further split into private and common messages). In order to satisfy the power constraint for each user, $\alpha_i(n)$ and $\delta(n)$ must satisfy:  $\sum_{n=1}^N \delta(n) \alpha_i(n) P = P$ for $i=1,2$. Clearly, orthogonal signaling is a special case of the general HK strategy (if $\alpha_1(n) = 0$ while $\alpha_2(n) \neq 0$, or vice versa).

We consider time sharing with $N=2$ time slots of equal durations (i.e., $\delta(1)=\delta(2)=\frac{1}{2}$). Thus, the power constraint per user $i$ is such that: $\alpha_i(1) + \alpha_i(2) = 2$ for $i=1,2$. It is noted that TDMA is a special case of this strategy (e.g. if $\alpha_1(1) = \alpha_2(2) = 0$ and $\alpha_1(2) = \alpha_2(1) = 2$). Since $\delta(1)=\delta(2)$, we assume that $\alpha_1(1)=\alpha_2(2)$, $\alpha_1(2)=\alpha_2(1)$, $\lambda_1(1)=\lambda_2(2)$ and $\lambda_1(2)=\lambda_2(1)$. This assumption guarantees equal rates for both users ($R_1=R_2$). Hence, by dropping the time index notation,
the maximum HK sum rate for this case can be obtained as a straightforward extension of Preposition 1:
\footnotesize
\begin{eqnarray}
	\nonumber \lefteqn{R_{\text{TS}}(a,P)  = \max_{0\leq \lambda_1,\lambda_2 \leq 1 , 0\leq \alpha_1,\alpha_2 \leq 2} (R_1 + R_2)} \\
  \nonumber  & = &\max_{0\leq \lambda_1,\lambda_2 \leq 1 , 0\leq \alpha_1,\alpha_2 \leq 2 } \bigg[ \gamma \bigg( \frac{\lambda_1 \alpha_1 P}{1+a\lambda_2 \alpha_2 P}\bigg) + \gamma \bigg( \frac{\lambda_2 \alpha_2 P}{1+a\lambda_1 \alpha_1 P}\bigg)\\
    \nonumber     & +&  \min \bigg\{  \gamma \bigg( \frac{a\bar{\lambda}_2 \alpha_2 P}{1+\lambda_1 \alpha_1 P+a\lambda_2 \alpha_2 P}\bigg)  + \gamma \nonumber  \bigg( \frac{a\bar{\lambda}_1 \alpha_1  P}{1+\lambda_2 \alpha_2 P+a\lambda_1 \alpha_1 P}\bigg) , \\
    \nonumber     & & \frac{1}{2} \gamma \bigg( \frac{\bar{\lambda}_1\alpha_1 P+a\bar{\lambda}_2\alpha_2 P}{1+\lambda_1 \alpha_1 P+a\lambda_2\alpha_2 P}\bigg) + \frac{1}{2} \gamma \bigg( \frac{\bar{\lambda}_2\alpha_2 P+a\bar{\lambda}_1 \alpha_1 P}{1+\lambda_2 \alpha_2 P+a\lambda_1\alpha_1 P}\bigg) \\
      & &  \gamma \bigg( \frac{\bar{\lambda}_1\alpha_1 P}{1+\lambda_1 \alpha_1P+a\lambda_2 \alpha_2 P}\bigg)  + \gamma \bigg( \frac{\bar{\lambda}_2 \alpha_2 P}{1+\lambda_2 \alpha_2 P+a\lambda_1 \alpha_1 P}\bigg)  \bigg\} \bigg]
\end{eqnarray}
\normalsize
such that $\alpha_1 + \alpha_2 = 2$.

After numerically solving this optimization problem, we reach that the maximum HK sum rate for this case is the maximum of the rates achieved using the following schemes:
\begin{enumerate}
  \item TDMA: $\alpha_1(1) = \alpha_2(2) = 0$ and $\alpha_1(2) = \alpha_2(1) = 2$, yielding the sum rate: $R_{orth}= \gamma(2P)$
  \item Set $\lambda_1(n)=\lambda_2(n)$ and $\alpha_1(n)=\alpha_2(n)=1$ (i.e., no advantage for using time sharing), yielding the sum rate $R_{sym}$ which is given in Theorem 1.
  \item Set $\lambda_1$ or  $\lambda_2$ to zero. Using time sharing gives a slight advantage over the asymmetric scheme rate given in (\ref{UnEqPwrRate}) (i.e., $R_{asymTS} \geq R_{asym}$). It is worth noting that if we choose to set $\lambda_1 = 0$, then the optimization results in  $\alpha_2 > \alpha_1$.
\end{enumerate}
Hence, we conjecture that $R_{\text{TS}}$ is the maximum of $R_{orth}$, $R_{sym}$, and $R_{asymTS}$, although this has yet to be proved.

It is also worth noting that a specific time sharing scheme was considered in \cite{Sason}. The author considered the case of $N=4$ time slots, with the following assumptions: $\delta(n)= \beta$ for $n=1,2$, $\alpha_i(n)=2 \beta$ for $i=1,2$ and $n=1,2$, $\delta(n)=\frac{1-2\beta}{2}$ for $n=3,4$, $\alpha_1(3)=\alpha_2(4)=2 (1+2\beta)$,  $\alpha_2(3)=\alpha_1(4)=0$,  $\lambda_1(1) = \lambda_2(2)$ and $\lambda_1(2) = \lambda_2(1)$. As an extension of Preposition 1, the maximum HK sum rate for this case is given by:
\begin{eqnarray}\label{RateSason}
	\nonumber  \lefteqn{R_{\text{Sason}}(a,P)  = \max_{0\leq \lambda_1,\lambda_2 \leq 1 , 0\leq \beta \leq \frac{1}{2}}  }\\
 & &  \bigg(  2 \beta  R_{RS}(a,2 \beta P) + (1-2\beta) \gamma\big( 2 (1+2\beta)P \big)  \bigg)
\end{eqnarray}
Allowing for time sharing with the same assumptions as in \cite{Sason} gives a slight advantage over the case of no time sharing as shown in Section \ref{numerical}.


\section{Numerical Results}\label{numerical}

\begin{figure}
\centering
\includegraphics[height=5.5cm, width=\columnwidth]{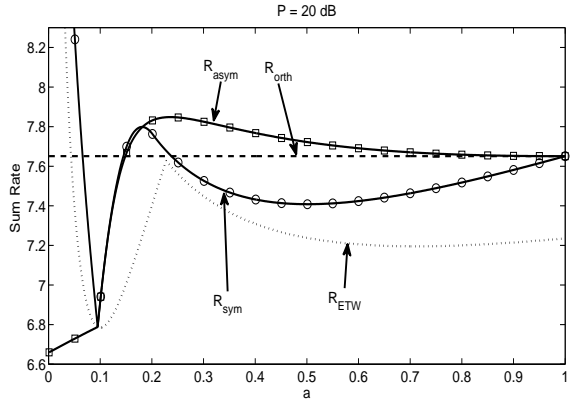}
\caption{Comparing $R_{sym}$, $R_{asym}$, $R_{ETW}$ and $R_{orth}$ for $P=20dB$.}
\label{Fixed_P}
\end{figure}

\begin{figure}
\centering
\includegraphics[height=5.5cm, width=\columnwidth]{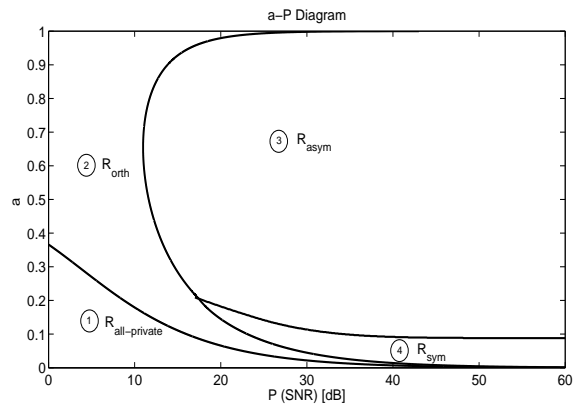}
\caption{Comparing $R_{sym}$, $R_{asym}$ and $R_{orth}$ for different values of $a$ and $P$.}
\label{a_P}
\end{figure}
\begin{figure}
\centering
\includegraphics[height=5.5cm, width=\columnwidth]{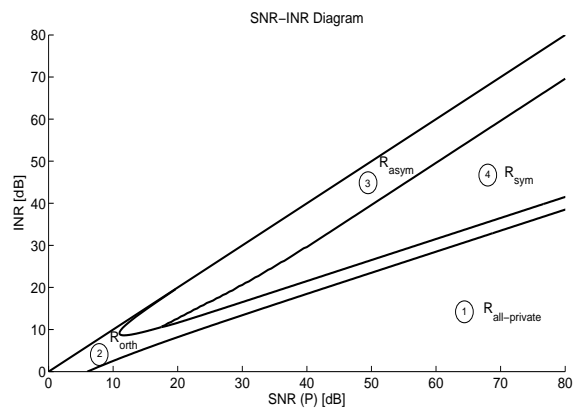}
\caption{Comparing $R_{sym}$, $R_{asym}$ and $R_{orth}$ for different SNR and INR values.}
\label{SNR_INR}
\end{figure}

In Figure \ref{Fixed_P}, the value of $P$ is fixed to $20$ dB, and rates corresponding to the different schemes with no time sharing ($R_{sym}$, $R_{asym}$, $R_{ETW}$ and $R_{orth}$) are plotted. For small values of $a$ ($a < 0.066$), symmetric rate splitting (private message only) achieves the largest rate.  After this there is a small region for which orthogonal is best ($0.066 \leq a < 0.145$), followed by a small region where symmetric rate splitting is again the best ($0.145\leq a < 0.182$).  Finally, for $0.182 \leq a < 0.9792$ asymmetric rate splitting achieves the largest rate while orthogonal is again the best for the remaining small region ($0.9792 \leq a < 1$). Notice also that $R_{ETW} \leq R_{sym}$ as expected, due to the sub-optimal choice of $\lambda$ in $R_{ETW}$.

Figure \ref{a_P} compares $R_{sym}$, $R_{asym}$ and $R_{orth}$ and shows which strategy is best at each value of $a$ and $P$. The numbered regions in the figure correspond to:
\begin{enumerate}
  \item Symmetric rate split with $\lambda_1=\lambda_2=1$ (i.e., private messages only). 
  \item Orthogonal signaling. 
  \item Asymmetric rate split. 
  \item Symmetric rate split.
\end{enumerate}
In order to understand these results from the perspective of \cite{Etkin}, the information in Figure \ref{a_P} is re-plotted in
Figure \ref{SNR_INR} with y-axis equal to $\text{INR}_{dB}= \log (aP)$ (in dB units) instead of $a$. We can see from Figure \ref{SNR_INR} that if $\frac{\text{INR}_{dB}}{\text{SNR}_{dB}} < \frac{1}{2}$ (i.e., region 1), both users should send only private messages. If  $\frac{\text{INR}_{dB}}{\text{SNR}_{dB}} \approx \frac{1}{2}$ (i.e., region 2), orthogonal signaling should be used. If $\frac{\text{INR}_{dB}}{\text{SNR}_{dB}} \approx 1$ (i.e., region 3), one of the users should send the common message only. Finally, if $\frac{1}{2}<\frac{\text{INR}_{dB}}{\text{SNR}_{dB}} < 1$ (i.e., region 4), both users should use the same private/common split ratio.

\begin{figure}
\centering
\includegraphics[height=5.5cm, width=\columnwidth]{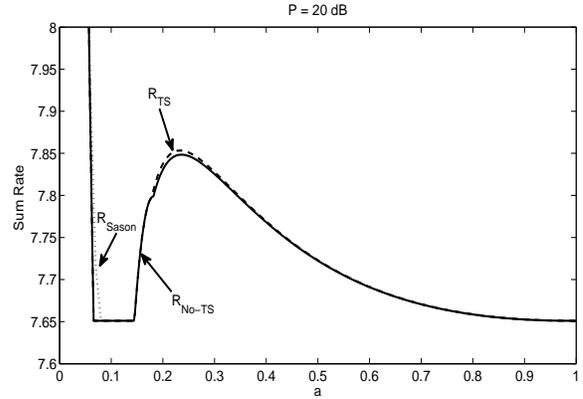}
\caption{Showing the slight advantage of the time sharing schemes $R_{\text{TS}}$ and $R_{\text{Sason}}$ over the case with no time sharing for $P=20dB$.}
\label{Fixed_P_TS}
\end{figure}

\begin{figure}
\centering
\includegraphics[height=5.5cm, width=\columnwidth]{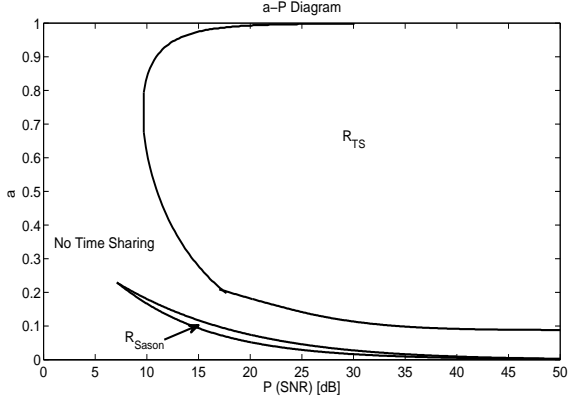}
\caption{Showing the regions in which time sharing schemes outperform the no time sharing case for different values of $a$ and $P$.}
\label{a_P_TS}
\end{figure}

In Figure \ref{Fixed_P_TS}, the value of $P$ is fixed to $20$ dB, and rates corresponding to the time sharing schemes $R_{\text{TS}}$ and $R_{\text{Sason}}$ are compared to the maximum rate achieved with no time sharing (i.e., maximum of $R_{sym}$, $R_{asym}$ and $R_{orth}$). The slight rate advantage of the time sharing schemes is apparat in the figure. Figure  \ref{a_P_TS} shows the values of $a$ and $P$ at which the rates achieved using the time sharing schemes $R_{\text{TS}}$ and $R_{\text{Sason}}$ outperforms the maximum rate achieved with no time sharing.


\section{Conclusions}
We studied the HK achievable sum rate for the two-user symmetric Gaussian interference channel. Without the allowance of time sharing we reached the following results: (a) we derived a closed form expression for the HK sum rate using symmetric power splits at both users and for the HK sum rate using asymmetric power splits achieved with one of the users sending only a common message; (b) we conjectured that the maximum HK sum-rate is achieved either using symmetric power splits or constraining one of the users to send only a common message (i.e., $R_{RS} = \max \{ R_{sym} ,  R_{asym} \}$); (c) we showed that the asymmetric rate outperforms the symmetric one for a wide range of $a$ and $P$ values. At the high SNR regime, we showed that for $a>0.087$, the rate achieved using the asymmetric power splitting outperforms the symmetric case; (d) we showed that orthogonal signaling performs good for a wide range of the low SNR regime and for $\frac{\text{INR}_{dB}}{\text{SNR}_{dB}} \approx \frac{1}{2}$. Finally, we considered specific time sharing schemes and we show that the advantage of using such time sharing schemes, as opposed to the case of not allowing time sharing, is quite small.


\appendices
\section{Proof of Theorem 1}
 Substituting $\lambda_1 = \lambda_2 = \lambda$ in equation \ref{SumRate}:

\begin{eqnarray}
	\nonumber R_{sym} & =&   \max_{0\leq \lambda \leq 1} \min \bigg\{ 2 \gamma \bigg( \frac{\lambda P + a \bar{\lambda} P}{1+a\lambda P}\bigg) \quad ,\\
    \nonumber   & & \qquad \qquad \gamma \bigg( \frac{\lambda P}{1+a\lambda P}\bigg) + \gamma \bigg( \frac{P+a\bar{\lambda} P}{1+a\lambda P}\bigg) \bigg\}\\
    & = &  \max_{0\leq \lambda \leq 1} \min \{ \Psi_1(a , \lambda , P) \, , \, \Psi_2(a , \lambda , P)\},
\end{eqnarray}
where $\Psi_1(a , \lambda , P) = 2 \gamma \bigg( \frac{\lambda P + a \bar{\lambda} P}{1+a\lambda P}\bigg)$
and  \\
$\Psi_2(a , \lambda , P) = \gamma \bigg( \frac{\lambda P}{1+a\lambda P}\bigg) + \gamma \bigg( \frac{P+a\bar{\lambda} P}{1+a\lambda P}\bigg)$.

We note that $R_{sym}$ is upper bounded as:
\begin{eqnarray}\label{upperBound_general}
   \lefteqn{R_{sym}  =   \max_{0\leq \lambda \leq 1} \min \{ \Psi_1(a , \lambda , P)  ,  \Psi_2(a , \lambda , P) \} }  \\
    &\leq &   \max_{0\leq \lambda \leq 1} \Psi_j(a , \lambda , P) =   \Psi_j(a , \lambda^{j*} , P)
\end{eqnarray}
for $j=1,2$, where $\lambda^{j*}(a,P) = \arg \max_{0\leq \lambda \leq 1} \Psi_j(a , \lambda , P)$.

$R_{sym}$ is also lower bounded by:
\begin{eqnarray}\label{lowerBound_general}
   R_{sym} & \geq &  \min \{ \Psi_1(a , \lambda^{j*} , P) \, , \, \Psi_2(a , \lambda^{j*} , P) \}
\end{eqnarray}
for $j=1,2$.

We first consider the case $P<\frac{1-a}{a^2}$.  It is straightforward to see
 that $\frac{\partial  \Psi_1(a , \lambda , P)}{\partial \lambda} > 0$ for this range of $P$, and thus
 $\lambda^{1*}(a,P)=1$. From (\ref{upperBound_general}), this implies $R_{sym} \leq \Psi_1(a , 1 , P)$.
 Since $\Psi_1(a , 1 , P) = \Psi_2(a , 1 , P)$ for any values of $a$ and $P$,
from (\ref{lowerBound_general}) we have $R_{sym} \geq \Psi_1(a , 1 , P)$.  Since the upper and lower
bounds match, we have shown the first case of (\ref{EqPwrRate}).

We next consider the range  $P>\frac{1-a}{a^2}$.
Because $\left(\frac{\partial  \Psi_2(a , \lambda , P)}{\partial \lambda} = 0\right)$ has only one solution at $\lambda = \frac{1-a}{(1+a)(a P)}$, and since $\frac{\partial  \Psi_2(a , \lambda , P)}{\partial \lambda} > 0$
for $\lambda < \frac{1-a}{(1+a)(a P)}$ while $\frac{\partial  \Psi_2(a , \lambda , P)}{\partial \lambda} < 0$
for $\lambda > \frac{1-a}{(1+a)(a P)}$, it follows that: $\lambda^{2*}(a,P)=\frac{1-a}{(1+a)(a P)}$.

We now  restrict attention to $P> \frac{1-a^3}{a^3(a+1)}$.  By (\ref{upperBound_general}),
$R_{sym} \leq \Psi_2(a , \lambda^{2*}, P)$.  For this range of $P$, it is straightforward to see that
$\Psi_2(a , \lambda^{2*} , P) < \Psi_1(a ,\lambda^{2*} , P)$.  Thus,
(\ref{lowerBound_general}) gives $R_{sym} \geq \Psi_2(a , \lambda^{2*}, P)$.  The upper and lower
bounds meet, thereby giving the third case in (\ref{EqPwrRate}).

We finally consider the remaining power region $\frac{1-a}{a^2} < P \leq \frac{1-a^3}{a^3(a+1)}$,
for which we note the following:
\begin{itemize}
    \item $\frac{\partial  \Psi_1(a , \lambda , P)}{\partial \lambda} < 0$
    \item $\Psi_2(a , \lambda^{2*} , P) > \Psi_1(a , \lambda^{2*} , P)$
    \item $\Psi_2(a , \lambda , P)$ is increasing in $\lambda$ for $\lambda < \lambda^{2*}$
\end{itemize}
As a result, it follows that the maximum occurs at the intersection of $\Psi_1(a , \lambda , P)$ and $\Psi_2(a , \lambda , P)$ (the intersection
occurs in the valid range).  The value of $\lambda$ at the intersection is $\lambda^*(a,P) =  \frac{a^2P+a-1}{P}$ which completes the proof of the final case in (\ref{EqPwrRate}). In the sum-rate expression in (\ref{EqPwrRate}), the function $\Psi_1(a,\lambda,P)$ is active for $P \leq \frac{1-a}{a^2}$, $\Psi_2(a,\lambda,P)$ is active for $P > \frac{1-a^3}{a^3(a+1)}$, and the remaining region corresponds to the intersection of $\Psi_1$ and $\Psi_2$.


\section{Asymmetric Rate Splitting Conjecture}
Equation \ref{SumRate} can be written as:
\begin{equation}
     R_{RS}(a,P) = \max_{0\leq \lambda_1,\lambda_2 \leq 1} \min (\Phi_1(a,P,\lambda_1,\lambda_2) , \Phi_2(a,P,\lambda_1,\lambda_2))
\end{equation}
where,
\begin{eqnarray}
  \nonumber   \lefteqn{\Phi_1(a,P,\lambda_1,\lambda_2)  =} \\
  \nonumber  & & \gamma \bigg( \frac{\lambda_1 P}{1+a\lambda_2 P}\bigg) + \gamma \bigg( \frac{\lambda_2 P}{1+a\lambda_1 P}\bigg) \\
  \nonumber  & & + \gamma \bigg( \frac{a\bar{\lambda}_2 P}{1+\lambda_1 P+a\lambda_2 P}\bigg) + \gamma \bigg( \frac{a\bar{\lambda}_1 P}{1+\lambda_2 P+a\lambda_1 P}\bigg) \\
    \nonumber  \lefteqn{\Phi_2(a,P,\lambda_1,\lambda_2) = }\\
 \nonumber   & & \gamma \bigg( \frac{\lambda_1 P}{1+a\lambda_2 P}\bigg) + \gamma \bigg( \frac{\lambda_2 P}{1+a\lambda_1 P}\bigg) \\
 \nonumber   & & + \frac{1}{2} \gamma \bigg( \frac{\bar{\lambda}_1 P+a\bar{\lambda}_2 P}{1+\lambda_1 P+a\lambda_2 P}\bigg) + \frac{1}{2} \gamma \bigg( \frac{\bar{\lambda}_2 P+a\bar{\lambda}_1 P}{1+\lambda_2 P+a\lambda_1 P}\bigg)
\end{eqnarray}

For $ P \geq \frac{1-a}{a^2} $ and for a fixed value of $\lambda_2$, we find that $\Phi_1$ is a decreasing function in $\lambda_1$ (i.e., $\frac{\partial \Phi_1 }{\partial \lambda_1} < 0$). The function $\Phi_2$ is either monotonically increasing, monotonically decreasing or increasing then decreasing in  $\lambda_1$. Specifically, fixing $\lambda_2$ we can observe the following:
\begin{enumerate}
  \item For small values of $\lambda_2$, $\Phi_2(\lambda_1)$ is monotonically increasing in $\lambda_1$, thus $\lambda_1$ that solves $\max_{0\leq \lambda_1, \leq 1} \min (\Phi_1(\lambda_1 ), \Phi_2(\lambda_1 ))$ is at the intersection of $\Phi_1(\lambda_1 )$ and  $\Phi_2(\lambda_1 )$ (i.e., at the value of $\lambda_1$ that satisfies $\Phi_1(\lambda_1 )=\Phi_2(\lambda_1 )$)
  \item  For large values of $\lambda_2$, $\Phi_2(\lambda_1)$ is monotonically decreasing, thus $\lambda_1$ that solves  $\max_{0\leq \lambda_1, \leq 1} \min (\Phi_1(\lambda_1 ), \Phi_2(\lambda_1 ))$ is at $\lambda_1 =0$.
    \item For the remaining values of $\lambda_2$, $\Phi_2(\lambda_1)$ is increasing till a certain value of $\lambda_1$ then decreasing, thus $\lambda_1$ that solves $\max_{0\leq \lambda_1, \leq 1} \min (\Phi_1(\lambda_1 ), \Phi_2(\lambda_1 ))$ is either at $\lambda_1$ that satisfies $\Phi_1(\lambda_1 )=\Phi_2(\lambda_1 )$ or at $\lambda_1= \arg \max_{\lambda_1}\Phi_1(\lambda_1 ) $.
\end{enumerate}
Therefore, based on these observations and from the symmetry of $\Phi_1$ and $\Phi_2$ with respect to $\lambda_1$ and $\lambda_2$, we conclude that $\lambda_1$ and $\lambda_2$ that solve $\max_{0\leq \lambda_1,\lambda_2 \leq 1} \min (\Phi_1(a,P,\lambda_1,\lambda_2) , \Phi_2(a,P,\lambda_1,\lambda_2))$ is either at the maximum point of the line of intersection between $\Phi_1$ and $\Phi_2$ or at the local maximum of $\Phi_1$. Two further observations; it can be shown that the maximum point of the intersection line is either at $\lambda_1=\lambda_2$ or at $\lambda_1=0$ (or $\lambda_2 =0$), and the local maximum of $\Phi_1$ takes place at $\lambda_1=\lambda_2$ (This can be shown by differentiating $\Phi_1$ with respect to $\lambda_1$ and $\lambda_2$ and simultaneously solving both equations for $\lambda_1$ and $\lambda_2$).

Hence, we conclude that the solution of $\max_{0\leq \lambda_1,\lambda_2 \leq 1} \min (\Phi_1(a,P,\lambda_1,\lambda_2) , \Phi_2(a,P,\lambda_1,\lambda_2))$ has two possibilities:
\begin{enumerate}
  \item Symmetric splitting ratio at both users (i.e., $\lambda_1 = \lambda_2 = \lambda_{sym}$)
  \item Asymmetric power split with $\lambda_1 = 0$ and $\lambda_2 = \lambda_{asym}$ (or vice versa).
\end{enumerate}

\end{document}